\documentclass[letterpaper,10pt,twocolumn,twoside,conference]{ieeeconf}
\IEEEoverridecommandlockouts                         
\usepackage{hyperref}
\usepackage{graphicx}		
\usepackage[format=plain,font=footnotesize,labelfont=bf,labelsep=period]{caption}
\usepackage{sidecap} 
\usepackage[export]{adjustbox}

\usepackage{amsmath}
\usepackage{amssymb}
\usepackage{mathtools}
\usepackage{algorithm}
\usepackage{algpseudocode}

\usepackage{pdfsync}
\usepackage[normalem]{ulem}
\usepackage{paralist}	
\usepackage[space]{grffile} 

\usepackage{color}
\usepackage{wrapfig}
\usepackage[format=plain,font=footnotesize,labelfont=bf,labelsep=period]{caption}
\usepackage{sidecap} 
\usepackage[export]{adjustbox}
\usepackage{dblfloatfix}

\usepackage{amsthm}

\newtheorem{theorem}{Theorem}

\newtheorem{lemma}{Lemma}
\theoremstyle{definition}
\newtheorem{definition}{Definition}
\theoremstyle{remark}
\newtheorem{remark}{Remark}
\theoremstyle{definition}
\newtheorem{assumption}{Assumption}
\theoremstyle{definition}

\newtheorem{example}{Example}

\newcommand{\R}{\mathbb{R}}

\definecolor{darkblue}{RGB}{0,0,102}
\definecolor{lightblue}{RGB}{77,77,148}

\definecolor{gold}{RGB}{234, 170, 0}
\definecolor{metallic_gold}{RGB}{139, 111, 78}

\renewcommand{\cal}[1]{\mathcal{ #1 }}
\newcommand{\mb}[1]{\mathbf{ #1 }}
\newcommand{\bs}[1]{\boldsymbol{ #1 }}

\newcommand{\grad}{\nabla}

\DeclareMathOperator*{\argmin}{argmin}

\newcommand{\certfn}{\mathtt{C}}
\newcommand{\fhat}{\widehat{\mb{f}}}
\newcommand{\ghat}{\widehat{\mb{g}}}
\newcommand{\certhat}{\widehat{\dot \certfn}}
\newcommand{\errf}{\widetilde{\mb{f}}}
\newcommand{\errg}{\widetilde{\mb{g}}}
\newcommand{\errF}{\widetilde{\mb{F}}}

\newcommand{\Lips}{\mathcal{L}}

\title{\LARGE \textbf{Towards Robust Data-Driven Control Synthesis for \\ Nonlinear Systems with Actuation Uncertainty}}
\author{Andrew J. Taylor$^{1}$, Victor D. Dorobantu$^{1}$, Sarah Dean$^{1}$, Benjamin Recht, Yisong Yue, and Aaron D. Ames
\thanks{$^1$Authors contributed equally. A.J. Taylor, V.D. Dorobantu, Y. Yue, and A.D. Ames are with the Department of Computing and Mathematical Sciences, California Institute of Technology, Pasadena, CA 91125, USA, {\tt\small \{ajtaylor, vdoroban, yyue, ames\}@caltech.edu}. S. Dean and B. Recht are with the Department of Electrical Engineering and Computer Sciences, University of California at Berkeley, Berkeley, CA 94720, USA, {\tt \small \{sarah,brecht\}@berkeley.edu}.
}
}

\begin{document}

\maketitle

\begin{abstract}%
Modern nonlinear control theory seeks to endow systems with properties such as stability and safety, and has been deployed successfully across various domains. Despite this success, model uncertainty remains a significant challenge in ensuring that model-based controllers transfer to real world systems. This paper develops a data-driven approach to robust control synthesis in the presence of model uncertainty using Control Certificate Functions (CCFs), resulting in a convex optimization based controller for achieving properties like stability and safety. An important benefit of our framework is nuanced data-dependent guarantees, which in principle can yield sample-efficient data collection approaches that need not fully determine the input-to-state relationship. This work serves as a starting point for addressing important questions at the intersection of nonlinear control theory and non-parametric learning, both theoretical and in application. We demonstrate the efficiency of the proposed method with respect to input data in simulation with an inverted pendulum in multiple experimental settings.
\end{abstract}

\section{Introduction}
Ensuring properties such as stability and safety is of significant importance in many modern control applications, including autonomous driving, industrial robotics, and aerospace vehicles. In practice, the models used to design these controllers are imperfect, with model uncertainty arising due to unmodeled dynamics and parametric errors. In the presence of such uncertainty, controllers may fail to render systems stable or safe. As real world control systems become increasingly complex, the potential for detrimental modeling errors increases, and thus it is critical to study control synthesis in the presence of uncertainty.

In this work, we propose a control synthesis process using \textit{control certificate functions} (CCFs) \cite{dimitrova2014deductive, boffi2020learning} that incorporates a data-driven approach for capturing model uncertainty. CCFs generalize popular tools from nonlinear control for achieving stability and safety such as Control Lyapunov Functions (CLFs) \cite{artstein1983stabilization}, Control Barrier Functions (CBFs) \cite{ames2014control}, and Control Barrier-Lyapunov Functions \cite{prajna2004safety}. CLFs and CBFs have been successfully deployed in the context of bipedal robotics \cite{ames2014rapidly, nguyen20163d}, adaptive cruise control \cite{ames2014control}, robotic manipulators \cite{khansari2014learning}, and multi-agent systems \cite{pickem2017robotarium}. Data-driven and machine learning based approaches have shown great promise for controlling systems with an uncertain model or with no model at all \cite{kober2013reinforcement,shi2019neural, cheng2019end,lee2020learning}. The integration of techniques from nonlinear control theory for achieving stability and safety with data-driven methods has become increasingly popular \cite{aswani2013provably, beckers2019stable, berkenkamp2015safe, gillula2012guaranteed, qu2020combining}, with many approaches relying on certificate functions for theoretical guarantees \cite{khansari2014learning, lederer2020parameter, choi2020reinforcement, cohen2020approximate, castaneda2020gaussian}.

Uncertainty in the effect of actuation remains a major challenge in achieving control-theoretic guarantees with data-driven methods. Many existing approaches assume certainty in how actuation enters the dynamics \cite{fisac2018general, umlauft2018uncertainty, zheng2020learning}, use structured controllers requiring strong characterizations of this uncertainty \cite{beckers2019stable, lederer2020training}, or require high coverage datasets with (nearly) complete characterizations of the input-to-state relationship \cite{berkenkamp2016safe2,berkenkamp2017safe}. In practice, collecting such data can be prohibitively costly or damage the system, suggesting a need for data-driven approaches that accommodate actuation uncertainty without requiring this complete characterization.

\textbf{Our contribution:} The contribution of this work is a novel approach for robust data-driven control synthesis via CCFs for control-affine systems with model uncertainty, including actuation uncertainty, which is broadly applicable in many real-world settings such as robotic \cite{murray1994mathematical} and automotive systems \cite{ioannou1993autonomous}. In Section \ref{sec:datadriven} we incorporate data into a convex optimization based control synthesis problem as affine inequality constraints which restrict possible model uncertainties. This enables the choice of robust control inputs over convex uncertainty sets. Rather than requiring a full characterization of how input enters the system, this approach utilizes the affine structure of CCF dynamics to choose inputs. This reduces the impact of actuation uncertainty on the evolution of the certificate function and allows for guarantees of stability and safety. In summary, our results show that good performance can be achieved when training data provides sufficiently dense coverage of the state space and targeted excitation in input directions. The proposed approach provides a unique perspective for unifying nonlinear control and non-parametric machine learning that is well positioned to study both theoretical and application oriented questions at this intersection.



\section{Background}
\label{sec:background}
This section provides a review of certificate-based nonlinear control synthesis and an overview of how model uncertainty impacts these synthesis methods. 

\subsection{Control Certificate Functions}
Consider a nonlinear control affine system given by:
\begin{equation}
\label{eqn:dynamics}
    \dot{\mb{x}} = \mb{f}(\mb{x}) + \mb{g}(\mb{x})\mb{u},
\end{equation}
where $\mb{x}\in\R^n$, $\mb{u}\in\R^m$, and $\mb{f}:\R^n\to\R^n$ and $\mb{g}:\R^n\to \R^{n\times m}$ are locally Lipschitz continuous on $\R^n$. Further assume that the origin is an equilibrium point of the uncontrolled system ($\mb{f}(\mb{0})=\mb{0}$). In this work we assume that $\mb{u}$ may be chosen unbounded as in \cite{khansari2014learning,jankovic2018robust}. Given a locally Lipschitz continuous state-feedback controller $\mb{k}:\R^n\to\R^m$, the closed-loop system dynamics are:
\begin{equation}
\label{eqn:cloop}
    \dot{\mb{x}} = \mb{f}_{\textrm{cl}}(\mb{x})\triangleq \mb{f}(\mb{x})+\mb{g}(\mb{x})\mb{k}(\mb{x}).
\end{equation}
The assumption on local Lipschitz continuity of $\mb{f}$, $\mb{g}$, and $\mb{k}$ implies that $\mb{f}_\textrm{cl}$ is locally Lipschitz continuous. Thus for any initial condition $\mb{x}_0 \triangleq \mb{x}(0) \in \R^n$ there exists a maximum time interval $I(\mb{x}_0) = [0, t_{\textrm{max}})$ such that $\mb{x}(t)$ is the unique solution to \eqref{eqn:cloop} on $I(\mb{x}_0)$ \cite{perko2013differential}. 

The qualitative behavior (such as stability or safety) of the the closed-loop system \eqref{eqn:cloop} can be certified via the notion of a continuously differentiable \textit{certificate function} $\certfn:\R^n\to\R$. Given a \textit{comparison function} $\alpha:\R\to\R$ (specific to the qualitative behavior of interest), certification 
is specified as an inequality on the derivative of the certificate function along solutions to the closed-loop system:
\begin{equation}
\label{eqn:certificate}
   \dot{\certfn}(\mb{x}) = \grad \certfn(\mb{x})^\top \mb{f}_{\textrm{cl}}(\mb{x}) \leq -\alpha(\certfn(\mb{x}))\:.
\end{equation}
Synthesis of controllers that satisfy \eqref{eqn:certificate} by design motivates the following definition of a \textit{Control Certificate Function}:

\begin{definition}[Control Certificate Function (CCF)]
A continuously differentiable function $\certfn:\R^n\to\R$ is a \textit{Control Certificate Function (CCF)} for \eqref{eqn:dynamics} with comparison function $\alpha:\R\to\R$ if for all $\mb{x}\in\R^n$: 
\begin{equation}
\label{eqn:ccf}
     \inf_{\mb{u}\in\R^m}  \grad \underbrace{\certfn(\mb{x})^\top \mb{f}(\mb{x})}_{L_{\mb{f}}\certfn(\mb{x})}+ \underbrace{\grad \certfn(\mb{x})^\top\mb{g}(\mb{x})}_{L_{\mb{g}}\certfn(\mb{x})}\mb{u}  \leq -\alpha(\certfn(\mb{x})).
\end{equation}
\end{definition}
\noindent where $L_{\mb{f}}\certfn:\R^n\to\R$ and $L_{\mb{g}}\certfn:\R^n\to\R^m$. The control-affine nature of the system dynamics are preserved by the CCF, such that the only component of the input that impacts the evolution of the certificate function lies in the direction of $L_{\mb{g}}\certfn(\mb{x})$.  Given a CCF $\certfn$ for \eqref{eqn:dynamics} and a corresponding comparison function $\alpha$, we define the point-wise set of all control values that satisfy the inequality in \eqref{eqn:ccf}:
\begin{equation}
\label{eqn:K_ccf}
    K_{\textrm{ccf}}(\mb{x}) \triangleq \left\{\mb{u}\in\R^m ~\left|~ L_\mb{f}\certfn(\mb{x})+L_\mb{g}\certfn(\mb{x})\mb{u}  \leq -\alpha(\certfn(\mb{x})) \right.\right\}.
\end{equation}
Any nominal locally Lipschitz continuous controller $\mb{k}_d:\R^n\to\R^m$ can be modified to take values in the set $K_\textrm{ccf}(\mb{x})$ via the certificate-critical CCF-QP:
\begin{align}
\label{eqn:CCF-QP}
\tag{CCF-QP}
\mb{k}(\mb{x}) =  \,\,\underset{\mb{u} \in \R^m}{\argmin}  &  \quad \frac{1}{2} \| \mb{u} -\mb{k}_d(\mb{x})\|_2^2  \\
\mathrm{s.t.} \quad & \quad\grad \certfn(\mb{x})^\top\left(\mb{f}(\mb{x})+\mb{g}(\mb{x})\mb{u}\right)  \leq -\alpha(\certfn(\mb{x})). \nonumber
\end{align}

Before providing particular examples of certificate functions useful in control synthesis, we review the following definitions.
We denote a continuous function $\alpha:[0,a)\to\R_+$, with $a>0$, as \textit{class $\cal{K}$} ($\alpha\in\cal{K}$) if $\alpha(0)=0$ and $\alpha$ is strictly monotonically increasing. If $a=\infty$ and $\lim_{r\to\infty}\alpha(r)=\infty$, then $\alpha$ is \textit{class $\cal{K}_\infty$} ($\alpha\in\cal{K}_\infty$). A continuous function $\alpha:(-b,a)\to\R$, with $a,b>0$, is \textit{extended class $\cal{K}$} ($\alpha\in\cal{K}_e$) if $\alpha(0)=0$ and $\alpha$ is strictly monotonically increasing. If $a,b=\infty$, $\lim_{r\to\infty}\alpha(r)=\infty$, and $\lim_{r\to-\infty}\alpha(r)=-\infty$, then $\alpha$ is \textit{extended class $\cal{K}_\infty$} ($\alpha\in\cal{K}_{\infty,e}$). Finally, we note that $c\in\R$ is referred to as a \textit{regular value} of a continuously differentiable function $h:\R^n\to\R$ if $h(\mb{x})=c\implies \grad h(\mb{x}) \neq\mb{0}_n$.

\begin{example}[Stability via Control Lyapunov Functions]
   In the context of stabilization to the origin, a control certificate function $V:\R^n\to\R$ with a class $\cal{K}$ comparison function $\alpha\in\cal{K}$ that satisfies:
   \begin{equation}
       \alpha_1(\Vert\mb{x}\Vert)\leq V(\mb{x}) \leq \alpha_2(\Vert\mb{x}\Vert),
   \end{equation}
   for $\alpha_1,\alpha_2\in\cal{K}$, is a \textit{Control Lyapunov Function (CLF)} \cite{artstein1983stabilization, sontag1989universal}, with stabilization to the origin achieved by controllers taking values in the point-wise set $K_{\textrm{ccf}}$ given by \eqref{eqn:K_ccf} \cite{ames2017control}.  
\end{example}

\begin{example}[Safety via Control Barrier Functions]
 In the context of safety, defined as forward invariance \cite{blanchini1999set} of a set $~\mathcal{S}$, a control certificate function $h:\R^n\to\R$ with $0$ a regular value and a comparison function $\alpha\in\mathcal K_{\infty,e}$ that satisfies:
   \begin{equation}
       \mb{x}\in\mathcal{S} \implies h(\mb{x}) \leq 0,
   \end{equation}
   is a \textit{Control Barrier Function (CBF)} \cite{ames2014control,ames2017control}, with safety of the set $\mathcal{S}$ achieved by controllers taking values in the point-wise set $K_{\textrm{ccf}}$ given by \eqref{eqn:K_ccf} \cite{ames2019control}. We adopt the opposite sign convention for $h$ so satisfying \eqref{eqn:ccf} guarantees safety.
\end{example}



\subsection{Model Uncertainty}

In practice, uncertainty in the system dynamics \eqref{eqn:dynamics} exists due to parametric error and unmodeled dynamics, such that the functions $\mb{f}$ and $\mb{g}$ are not precisely known. Control affine systems are a natural setting to study actuation uncertainty as the function $\mb{g}$ can be seen as an uncertain gain multiplying the input. In this context, control synthesis is done with a nominal model that estimates the true system dynamics:
\begin{equation}
\label{eqn:nom_dynamics}
    \widehat{\dot{\mb{x}}} = \widehat{\mb{f}}(\mb{x}) + \widehat{\mb{g}}(\mb{x})\mb{u},
\end{equation}
where $\widehat{\mb{f}}:\R^n\to\R^n$ and $\widehat{\mb{g}}:\R^n\to\R^{n\times m}$ are locally Lipschitz continuous. Adding and subtracting this expression to and from \eqref{eqn:dynamics} implies the system evolution is described by:
\begin{equation}
\label{eqn:res_dynamics}
    \dot{\mb{x}} = \widehat{\mb{f}}(\mb{x}) + \widehat{\mb{g}}(\mb{x})\mb{u} + \underbrace{\mb{f}(\mb{x})-\widehat{\mb{f}}(\mb{x})}_{\widetilde{\mb{f}}(\mb{x})} + \underbrace{\left(\mb{g}(\mb{x})-\widehat{\mb{g}}(\mb{x})\right)}_{\widetilde{\mb{g}}(\mb{x})}\mb{u},
\end{equation}
where $\widetilde{\mb{f}}:\R^n\to\R^n$ and $\widetilde{\mb{g}}:\R^n\to\R^{n\times m}$ are the unmodeled dynamics. This uncertainty in the dynamics additionally manifests in the time derivative of a CCF for the system:
\begin{align}
\label{eqn:res_cerfn}
    \dot{\certfn}(\mb{x},\mb{u}) = & \overbrace{\underbrace{\grad \certfn(\mb{x})^\top\widehat{\mb{f}}(\mb{x})}_{L_{\widehat{\mb{f}}}\certfn(\mb{x})}+\underbrace{\grad \certfn(\mb{x})^\top\widehat{\mb{g}}(\mb{x})}_{L_{\widehat{\mb{g}}}\certfn(\mb{x})}\mb{u}}^{\widehat{\dot{\certfn}}(\mb{x},\mb{u})} \nonumber \\ & + \underbrace{\grad \certfn(\mb{x})^\top\widetilde{\mb{f}}(\mb{x})}_{L_{\errf}\certfn(\mb{x})} + \underbrace{\grad \certfn(\mb{x})^\top\widetilde{\mb{g}}(\mb{x})}_{L_{\errg}\certfn(\mb{x})}\mb{u},
\end{align}
where $L_{\widehat{\mb{f}}}\certfn, L_{\errf}\certfn:\R^n\to\R$, and $L_{\widehat{\mb{g}}}\certfn, L_{\errg}\certfn:\R^n\to\R^m$. The presence of uncertainty in the CCF time derivative makes it impossible to verify whether a given control input is in the set $K_{\textrm{ccf}}(\mb{x})$ given in \eqref{eqn:K_ccf}, and can lead to failure to achieve the desired qualitative behavior.

\begin{assumption}
\label{ass:wellposed}
The function $\certfn:\R^n\to\R$ is a valid CCF with comparison function $\alpha:\R\to\R$ for the true dynamic system \eqref{eqn:res_dynamics}.
Mathematically this assumption appears as:
\begin{equation*}
   \inf_{\mb{u}\in\R^m}\grad \certfn(\mb{x})^\top\left(\mb{f}(\mb{x})+\mb{g}(\mb{x})\mb{u}\right)\leq-\alpha(\certfn(\mb{x})).
\end{equation*}
\end{assumption}
This assumption is structural in nature and can be met for feedback linearizable systems (such as robotic systems).

\begin{remark}
Though many approaches to CCF design (linearization, energy-based, numerical sums-of-squares methods) rely on knowledge of the true dynamics, 
it is also possible to design CCFs without explicit knowledge of the true system dynamics.
For example, a valid CLF and comparison function for the true system can be designed via feedback linearization assuming only knowledge of the degree of actuation (see \cite{taylor2019episodic} for full details). 
This method also works to specify CBFs which are defined by sublevel sets of CLFs (as in our simulation results).
We emphasize the difference between choosing a qualitative behavior that the system can be made to satisfy (e.g. the CCF) and actually designing the control inputs which achieve the behavior. Our work focuses on the latter: solving the problem of choosing stable/safe inputs in the presence of uncertainty.
\end{remark}

\begin{assumption}
The functions $\widetilde{\mb{f}}$ and $\widetilde{\mb{g}}$ are globally Lipschitz continuous with known Lipschitz constants $\mathcal{L}_{\widetilde{\mb{f}}}$ and $\mathcal{L}_{\widetilde{\mb{g}}}$.
\end{assumption}

\begin{remark}
Knowledge of minimal Lipschitz constants is not necessary, but smaller Lipschitz constants are associated with improved performance of data-driven robust control methods.
\end{remark}

\section{Data-Driven Robust Control Synthesis}
\label{sec:datadriven}
In this section we explore how data can be incorporated directly into an optimization based controller to robustly achieve a desired qualitative behavior specified via a CCF.

Consider a dataset consisting of $N$ tuples of states, inputs, and corresponding state time derivatives, $\mathfrak{D} = \{(\mb{x}_i,\mb{u}_i,\dot{\mb{x}}_i)\}_{i=1}^N$,
with $\mb{x}_i\in\R^n$, $\mb{u}_i\in\R^m$, and $\dot{\mb{x}}_i\in\R^n$ for $i=1,\ldots, N$. It may not be possible to directly measure the state time derivatives $\dot{\mb{x}}_i$, but they can approximated from sequential state observations $\mb{x}_i$ \cite{taylor2019episodic,taylor2020learning}. While we do not consider noise, our construction can be modified to account for bounded noise at the expense of a constant offset to~\eqref{eqn:errbnd}.

We now show how data allows us to reduce uncertainty by constraining the possible values of the functions $\widetilde{\mb{f}}$ and $\widetilde{\mb{g}}$ directly, without a parametric estimator.
Considering the uncertain model \eqref{eqn:res_dynamics} evaluated at a state and input pair $(\mb{x}_i,\mb{u}_i)$ in the dataset yields:
\begin{equation}
\label{eqn:Ferr}
    \widetilde{\mb{F}}_i \triangleq   \dot{\mb{x}}_i - (\fhat(\mb{x}_i) + \ghat(\mb{x}_i)\mb{u}_i) = \errf(\mb{x}_i) + \errg(\mb{x}_i)\mb{u}_i,
\end{equation}
where $\widetilde{\mb{F}}_i\in\R^n$ can be interpreted as the error between the true state time derivative and the nominal model \eqref{eqn:nom_dynamics} evaluated at the state and input pair $(\mb{x}_i,\mb{u}_i)$.

Considering a state $\mb{x}\in\R^n$ (not necessarily present in the dataset $\mathfrak{D}$), the second equality in \eqref{eqn:Ferr} implies:
\begin{equation}
    \errf(\mb{x})+\errg(\mb{x})\mb{u}_i-\errF_i = \errf(\mb{x}) - \errf(\mb{x}_i) + \left(\errg(\mb{x}) - \errg(\mb{x}_i) \right)\mb{u}_i.
\end{equation}
This expression provides a relationship between the possible values of the unmodeled dynamics $\errf$ and $\errg$ at the state $\mb{x}$ and the values of the unmodeled dynamics at the data point $\mb{x}_i$. Using the local Lipschitz continuity of the unmodeled dynamics yields the following bound:
\begin{align}
\label{eqn:errbnd}
    \Big\Vert \errf(\mb{x}) + \errg(\mb{x}) & \mb{u}_i - \errF_i \Big\Vert_2 \nonumber \\ 
    &= \left\Vert \errf(\mb{x}) - \errf(\mb{x}_i) + \left( \errg(\mb{x}) - \errg(\mb{x}_i) \right)\mb{u}_i \right\Vert_2, \nonumber\\
    &\leq \left( \Lips_{\errf} + \Lips_{\errg} \Vert\mb{u}_i\Vert_2 \right) \Vert \mb{x} - \mb{x}_i \Vert_2 \triangleq \epsilon_i(\mb{x}).
\end{align}
where $\epsilon_i:\R^n\to\R_+$ for $i=1,\ldots,N$. We see that the bound grows with the magnitude of the Lipschitz constants $\Lips_{\errf}$ and $\Lips_{\errg}$ and distance of the state $\mb{x}$ from the data point $\mb{x}_i$. The values of $\Lips_{\errf}$ and $\Lips_{\errg}$ are not explicitly data dependent, and thus the bound can be improved for a given dataset by reducing the possible model uncertainty through improved modeling. Given this construction we may define the point-wise uncertainty set:
\begin{align}
\label{eqn:uncertainelementset}
    \mathcal{U}_i(\mb{x}) \triangleq &\bigg \{ (\mb{A},\mb{b})\in\R^{n\times m}\times\R^{n} ~\bigg|~ \nonumber\\ &\left\Vert \mb{b} +\mb{A}\mb{u}_i - \errF_i\right\Vert_2\leq  \epsilon_i(\mb{x})\bigg\}\subset\R^{n\times m}\times\R^{n},
\end{align}
noting that $(\errg(\mb{x}),\errf(\mb{x}))\in\mathcal{U}_i(\mb{x})$ and $\mathcal{U}_i(\mb{x})$ is closed and convex. Considering this construction over the entire dataset $\mathfrak{D}$ yields the following point-wise uncertainty set:
\begin{equation}
\label{eqn:uncertainset}
    \mathcal{U}(\mb{x}) \triangleq \bigcap_{i=1}^N \mathcal{U}_i(\mb{x})\subset\R^{n\times m}\times\R^{n},
\end{equation}
noting that $(\errg(\mb{x}),\errf(\mb{x}))\in\mathcal{U}(\mb{x})$ and $\mathcal{U}(\mb{x})$ is closed and convex. Therefore, $\mathcal{U}(\mb{x})$ consists of all possible model errors that are consistent with the observed data. This allows us to pose the following data robust control problem:
\begin{definition}[Data Robust Control Certificate Function Optimization Problem]
\begin{align}
\label{eqn:robust_general}
\tag{DR-CCF-OP}
    \mb{k}_\mathrm{rob}(\mb{x}) = &~\argmin_{\mb{u} \in \R^m}  \frac{1}{2}\Vert \mb{u} -\mb{k}_d(\mb{x})\Vert_2^2 \nonumber \\
    &~ \mathrm{s.t.}~ \certhat(\mb{x},\mb{u}) + \grad \certfn(\mb{x})^\top\left(\mb{b} + \mb{A}\mb{u}\right) \leq-\alpha(\certfn(\mb{x})) \nonumber \\ &\qquad\textrm{for all}~ (\mb{A},\mb{b}) \in \mathcal{U}(\mb{x}). \nonumber
\end{align}
\end{definition}
By construction we have that $(\errg(\mb{x}),\errf(\mb{x}))\in\mathcal{U}(\mb{x})$, implying that $\mb{k}_{\textrm{rob}}(\mb{x})\in K_{\textrm{ccf}}(\mb{x})$ when the problem is feasible. Thus the closed-loop system \eqref{eqn:cloop} under $\mb{k}_{\textrm{rob}}$ satisfies inequality \eqref{eqn:certificate}. We next present one of our main results, using robust optimization \cite{ben2009robust} to yield a convex problem for synthesizing such a robust controller.

\begin{theorem}[Robust Control Synthesis]
\label{thm:robustcontrolsynthesis}
Let $\certfn:\R^n\to\R$ be a control certificate with comparison function $\alpha:\R\to\R$. The robust controller \eqref{eqn:robust_general} is equivalently expressed as:
\begin{align}
\label{eqn:robust_final}
\tag{DR-CCF-SOCP}
    \mb{k}_{\mathrm{rob}}(\mb{x}) & =  \argmin_{\substack{\mb{u} \in \R^m  \\ \bs{\lambda}_i \in \R^n}} ~ \frac{1}{2}\Vert \mb{u}-\mb{k}_d(\mb{x})\Vert_2^2 \nonumber \\
    \mathrm{s.t.} &~ \certhat(\mb{x},\mb{u}) -\sum_{i=1}^N\left(\bs{\lambda}_i^\top \errF_i - \Vert\bs{\lambda}_i \Vert_2 \epsilon_i(\mb{x})\right) \leq -\alpha(\certfn(\mb{x})), \nonumber \\
    &~ \sum_{i=1}^N \bs{\lambda}_i \mb{u}_i^\top = -\grad \certfn(\mb{x})\mb{u}^\top, \quad \sum_{i=1}^N \bs{\lambda}_i= -\grad \certfn(\mb{x}).\nonumber
\end{align}
\end{theorem}

\begin{proof}
An input $\mb{u} \in \R^m$ is feasible if the optimal value of the optimization problem:
\begin{align*}
    \sup_{\substack{\mb{A} \in \R^{n \times m} \\ \mb{b} \in \R^n }} &~ \widehat{\dot{\certfn}}(\mb{x}, \mb{u}) + \grad \certfn(\mb{x})^\top(\mb{b}+\mb{Au})\\
    \mathrm{s.t.}~~ \Vert & \mb{b}+\mb{Au}_i - \errF_i \Vert_2 \leq \epsilon_i(\mb{x} )~\mathrm{for~all}~ i\in 1,\ldots,N,
\end{align*}
is less than or equal to $-\alpha(\certfn(\mb{x}))$. Each inequality constraint can be rewritten as set membership in a second-order cone. That is, $(\mb{b} + \mb{A}\mb{u}_i - \widetilde{\mb{F}}_i, \epsilon_i(\mb{x})) \in \mathcal{Q}_n$ for each $i \in \{ 1, \dots, N \}$, with $\mathcal{Q}_n \subset \R^{n + 1}$ denoting the Lorentz cone $\mathcal{Q}_n = \left\{ (\mb{y}, t) \in \R^n \times \R: \Vert\mb{y}\Vert_2 \leq t \right\}$. This yields the following dual problem:
\begin{align*}
    \inf_{\substack{\bs{\lambda}_i \in \R^n \\ \nu_i \in \R }} &~ \widehat{\dot{\certfn}}(\mb{x}, \mb{u}) -\sum_{i=1}^N\left( \bs{\lambda}_i^\top \errF_i - \nu_i \epsilon_i(\mb{x})\right)\\
    \mathrm{s.t.} &~ \sum_{i=1}^N \bs{\lambda}_i\mb{u}_i^\top = -\grad \certfn(\mb{x})\mb{u}^\top, \quad
     \sum_{i=1}^N \bs{\lambda}_i = -\grad \certfn(\mb{x}) \nonumber\\
    &~ \Vert \bs{\lambda}_i \Vert_2 \leq \nu_i ~\mathrm{for~all}~ i\in 1,\ldots,N,
\end{align*}
where $(\bs{\lambda}_i, \nu_i) \in \mathcal{Q}_n$ are Lagrange multipliers corresponding to the constraints imposed by the data point $(\mb{x}_i, \mb{u}_i)$. For any $(\mb{x}_i, \mb{u}_i)$, choosing Lagrange multipliers such that $\nu_i > \Vert \bs{\lambda}_i\Vert_2$ increases the value of the dual problem compared to choosing $\nu_i = \Vert\bs{\lambda}_i\Vert_2$. We therefore simplify the problem as:
\begin{align*}
    \inf_{\bs{\lambda}_i \in \R^n } &~ \widehat{\dot{\certfn}}(\mb{x}, \mb{u}) -\sum_{i=1}\left( \bs{\lambda}_i^\top \errF_i - \Vert \bs{\lambda}_i\Vert_2 \epsilon_i(\mb{x})\right)\\
    \mathrm{s.t.} &~ \sum_{i=1}^N \bs{\lambda}_i \mb{u}_i^\top = -\grad \certfn(\mb{x})\mb{u}^\top, \quad
     \sum_{i=1}^N \bs{\lambda}_i = -\grad \certfn(\mb{x}).
\end{align*}
This optimization problem can then replace the original robust constraint in \eqref{eqn:robust_general} to yield the final optimization problem \eqref{eqn:robust_final}.
\end{proof}

\begin{remark}
Solving this optimization problem in a computationally efficient manner may require data segmentation for higher-dimensional problems \cite{lederer2020real}, but a detailed consideration is outside the scope of this work, which is focused on theoretical foundations. 
\end{remark}
\section{Feasibility}
\label{sec:analysis}
In this section we provide an analysis of the feasibility of the controller proposed in Section \ref{sec:datadriven}. The feasibility of the controller \eqref{eqn:robust_final} at a given state $\mb{x}$ is determined by the structure of the uncertainty set $\mathcal{U}(\mb{x})$ defined in \eqref{eqn:uncertainset}. The following lemma provides a condition on the inputs in the dataset that implies $\mathcal{U}(\mb{x})$ is bounded for all $\mb{x}\in\R^n$:

\begin{lemma}[Bounded Uncertainty Sets]
\label{lem:boundeduncertainty}
Consider a dataset $\mathfrak{D}$ with $N$ data points satisfying $N\geq m+1$. If there exists a set of data points $\{(\mb{x}_i,\mb{u}_i,\dot{\mb{x}}_i)\}_{i=1}^{m+1}\subseteq \mathfrak{D}$ such that the set of vectors:
\begin{equation}
\mathcal{M} \triangleq \left\{\begin{bmatrix}\mb{u}_i^\top & 1\end{bmatrix}^\top\right\}_{i=1}^{m+1},
\end{equation}
are linearly independent, then the uncertainty set $\mathcal{U}(\mb{x})$ is bounded (and thus compact) for any $\mb{x}\in\R^n$. 
\end{lemma}
\begin{proof}
Consider the set $\mathcal{M}$, and define:
\begin{equation*}
\mathcal{M}_{\mb{u}} \triangleq \left\{\mb{u}_i\in\R^m ~\bigg{|}~ \begin{bmatrix}\mb{u}_i \\ 1 \end{bmatrix}\in\mathcal{M} \right\}.
\end{equation*}
For arbitrary $\mb{x}\in\R^n$, let $(\mb{A},\mb{b})\in\mathcal{U}(\mb{x})$. By definition of the individual uncertainty sets $\mathcal{U}_i(\mb{x})$ and $\mathcal{U}(\mb{x})$ given in \eqref{eqn:uncertainelementset} and \eqref{eqn:uncertainset}, respectively, we have that:
\begin{equation}
\label{eqn:uncertainbound}
\left\Vert\mb{Au}_i+\mb{b}-\widetilde{\mb{F}}_i\right\Vert_2 \leq \epsilon_i(\mb{x}),
\end{equation}
for $\mb{u}_i\in\mathcal{M}_{\mb{u}}$. Defining $\mb{v}_i \triangleq \mb{Au}_i+\mb{b}-\widetilde{\mb{F}}_i$, we have that:
\begin{equation*}
    \sum_{i=1}^{m+1}\left\Vert\mb{Au}_i+\mb{b}-\widetilde{\mb{F}}_i\right\Vert_2^2= \sum_{i=1}^{m+1}\sum_{j=1}^{n} \vert V_{ji} \vert^2 = \Vert \mb{V} \Vert_F^2,
\end{equation*}
where $\Vert\cdot\Vert_F$ is the Frobenius norm and $\mb{V} = \begin{bmatrix} \mb{v}_1 & \cdots & \mb{v}_{m+1} \end{bmatrix}$. Noting that $\Vert \mb{V} \Vert_F=\Vert \mb{V}^\top \Vert_F$, factoring and using the fact $\Vert\mb{P}\Vert_2 \leq \Vert\mb{P}\Vert_F$ for any $\mb{P}\in\R^{m+1 \times n}$ in conjunction with \eqref{eqn:uncertainbound} yields:
\begin{equation*}
     \left\Vert \begin{bmatrix} \mb{u}_1^\top & 1 \\ \vdots & \vdots \\  \mb{u}_{m+1}^\top & 1 \end{bmatrix}\begin{bmatrix} \mb{A}^\top \\ \mb{b}^\top \end{bmatrix} - \begin{bmatrix} \errF_1^\top \\ \vdots \\ \errF_{m+1}^\top \end{bmatrix} \right\Vert_2^2 \leq \sum_{i=1}^{m+1}\epsilon_i(\mb{x})^2.
\end{equation*}
Taking the square root of both sides and employing the reverse triangle inequality we arrive at:
\begin{equation*}
     \left\Vert \underbrace{\begin{bmatrix} \mb{u}_1^\top & 1 \\ \vdots & \vdots \\  \mb{u}_{m+1}^\top & 1 \end{bmatrix}}_{\mb{U}}\begin{bmatrix} \mb{A}^\top \\ \mb{b}^\top \end{bmatrix} \right\Vert_2 \leq \sqrt{\sum_{i=1}^{m+1}\epsilon_i(\mb{x})^2} + \left\Vert \begin{bmatrix} \errF_1^\top \\ \vdots \\ \errF_{m+1}^\top \end{bmatrix} \right\Vert_2.
\end{equation*}
Noting that $\mb{U}$ has full rank by the linear independence of the vectors in the set $\mathcal{M}$, we have that:
\begin{equation*}
    \sigma_{m+1}(\mb{U})\left\Vert \begin{bmatrix} \mb{A}^\top \\ \mb{b}^\top \end{bmatrix} \right\Vert_2 \leq \left\Vert \begin{bmatrix} \mb{u}_1^\top & 1 \\ \vdots & \vdots \\  \mb{u}_{m+1}^\top & 1 \end{bmatrix}\begin{bmatrix} \mb{A}^\top \\ \mb{b}^\top \end{bmatrix} \right\Vert_2,
\end{equation*}
where $\sigma_{m+1}(\mb{U})>0$ is the smallest singular value of $\mb{U}$. Combining the two previous inequalities shows that the uncertainty set $\mathcal{U}(\mb{x})$ is bounded.
\end{proof}

This result shows that variety in input directions is sufficient to assert boundedness of the uncertainty set as a uniform property over the entire state space. As seen in the proof, the bound on this set may be very large if the values of $\epsilon_i(\mb{x})$ are large (as is the case when $\mb{x}$ is far away from the data points $\mb{x}_i$ associated with the inputs used to construct the bound). Alternatively, the uncertainty set will be small for a point $\mathbf{x}$ if there is local variety in input directions within the training dataset.
To understand how the size of $\mathcal{U}(\mb{x})$ impacts feasibility of the optimization problem, we define the following set:
\begin{align}
\label{eqn:uncertainlie}
\widetilde{\mathcal{U}}_\certfn(\mb{x}) \triangleq \Big\{(&\mb{a},b)\in\R^m\times\R ~|~ \exists~(\mb{A},\mb{b})\in\mathcal{U}(\mb{x})~\textrm{s.t.}~  \nonumber \\ &\mb{a} = (\grad\certfn(\mb{x})^\top\mb{A})^\top,~b =\grad\certfn(\mb{x})^\top\mb{b}  \Big\},  
\end{align}
The set $\widetilde{\mathcal{U}}_\certfn(\mb{x})$ can be interpreted as the projection of the dynamics uncertainty set $\mathcal{U}(\mb{x})$ along the gradient of the control certificate function $\certfn$, creating an $m+1$ dimensional set which represents the possible uncertainties in the \textit{Lie derivatives} of $\certfn$. Additionally define the set:
\begin{equation}
\mathcal{U}_\certfn(\mb{x}) \triangleq \left \{ \left( L_{\widehat{\mb{g}}}\certfn(\mb{x})^\top,L_{\widehat{\mb{f}}}\certfn(\mb{x}) \right) \right\} \oplus \widetilde{\mathcal{U}}_\certfn(\mb{x}),
\label{eqn:uncertainliecentered}
\end{equation}
where $\oplus$ denotes a Minkowski sum. The set $\mathcal{U}_\certfn(\mb{x})$ is the recentering of $\widetilde{\mathcal{U}}_\certfn(\mb{x})$ set around the estimate of the Lie derivatives $(L_{\widehat{\mb{g}}}\certfn(\mb{x})^\top,L_{\widehat{\mb{f}}}\certfn(\mb{x}))\in\R^{m+1}$, such that it captures the possible true Lie derivatives of $\certfn$. As multiplication by $\grad\certfn(\mb{x})^\top$ is a linear transformation, $\widetilde{\mathcal{U}}_\certfn(\mb{x})$ and $\mathcal{U}_\certfn(\mb{x})$ are convex, and if $\mathcal{U}(\mb{x})$ is bounded (and therefore compact), then $\widetilde{\mathcal{U}}_\certfn(\mb{x})$ and $\mathcal{U}_\certfn(\mb{x})$ are compact. 

We present our second main result in the form of a necessary and sufficient condition for feasibility of \eqref{eqn:robust_general}:

\begin{theorem}[Feasibility of Data-Driven Robust Controller]
\label{thm:feasible}
For a state $\mb{x}\in\R^n$, let the sets $\mathcal{U}(\mb{x})$ and $\mathcal{U}_\certfn(\mb{x})$ be defined as in \eqref{eqn:uncertainset} and \eqref{eqn:uncertainliecentered}, respectively. Define the ray $\mathcal{R}\subset\R^{m+1}$ as $\mathcal{R} = \{\mb{0}_m\} \times (-\alpha(\certfn(\mb{x})),\infty)$. Assuming that $~\mathcal{U}(\mb{x})$ is bounded, the data-driven robust controller \eqref{eqn:robust_final} is feasible if and only if:
\begin{equation}
    \mathcal{U}_\certfn(\mb{x}) \cap \mathcal{R} = \emptyset.
\end{equation}
\end{theorem}
Intuitively, the ray $\mathcal{R}$ represents Lie derivative pairs that do not satisfy the certificate function condition with no actuation and cannot be modified through actuation. By Assumption \ref{ass:wellposed}, the Lie derivatives of the true system are not contained in $\mathcal{R}$, but the possible uncertainties permitted by data need not necessarily reflect this. If one possible uncertainty pair in the set $~\mathcal{U}_\certfn(\mb{x})$ is contained in $\mathcal{R}$, it is impossible to meet the certificate condition for that uncertainty pair.

\begin{proof}
The proof proceeds from the original structure of the robust control problem \eqref{eqn:robust_general}. The constraint on the input $\mb{u}^*$ specified by this controller is given by:
\begin{equation*}
    L_{\widehat{\mb{f}}}\certfn(\mb{x})+\grad\certfn(\mb{x})^\top\mb{b} + (L_{\widehat{\mb{g}}}\certfn(\mb{x})+ \grad\certfn(\mb{x})^\top\mb{A})\mb{u}^* \leq -\alpha(\certfn(\mb{x})).
\end{equation*}
for all $(\mb{A},\mb{b})\in\mathcal{U}(\mb{x})$. Given the definitions of $\widetilde{\mathcal{U}}_\certfn(\mb{x})$ and $\mathcal{U}_\certfn(\mb{x})$ in \eqref{eqn:uncertainlie} and \eqref{eqn:uncertainliecentered}, this constraint can be expressed as:
\begin{equation}
\label{eqn:simpleconstraint}
    q+\mb{p}^\top\mb{u}^* \leq -\alpha(\certfn(\mb{x}))
\end{equation}
for all $(\mb{p},q)\in\mathcal{U}_\certfn(\mb{x})$.

\subsubsection{Necessity}
Assume that $\mathcal{U}_\certfn(\mb{x}) \cap \mathcal{R} \neq \emptyset$. This implies that there exists a $(\mb{p},q)\in\mathcal{U}_\certfn(\mb{x})$ such that $q > -\alpha(\certfn(\mb{x}))$ and $\mb{p}=\mb{0}_m$. Thus for any input $\mb{u}\in\R^m$, we have that:
\begin{equation*}
    q + \mb{p}^\top\mb{u} = q > -\alpha(\certfn(\mb{x})),
\end{equation*}
violating the constraint in \eqref{eqn:simpleconstraint}. Thus the optimization problem is infeasible.

\subsubsection{Sufficiency}
Begin by defining the hyperplane $\mathcal{H}\subset\R^{m+1}$ with unit normal $\mb{n} = \begin{bmatrix} \mb{0}_m^\top & 1 \end{bmatrix}^\top$ and offset $-\alpha(\certfn(\mb{x}))$ and define the set:
\begin{equation*}
    \overline{\mathcal{U}}_\certfn(\mb{x}) = \left\{(\mb{p},q)\in\mathcal{U}_\certfn(\mb{x}) ~\big | ~ \left\langle \begin{bmatrix} \mb{0}_m \\ 1 \end{bmatrix},\begin{bmatrix} \mb{p} \\ q \end{bmatrix} \right\rangle \geq -\alpha(\certfn(\mb{x})) \right\},
\end{equation*}
which corresponds to the Lie derivative pairs in the set $~\mathcal{U}_\certfn(\mb{x})$ that do not meet the certificate function condition \eqref{eqn:simpleconstraint} strictly under no input ($\mb{u}^*=\mb{0}_m$). 
Note that if $\mb{u}^*=\mb{0}_m$ satisfies \eqref{eqn:simpleconstraint} for all $(\mb{p},q)\in\overline{\mathcal{U}}_\certfn(\mb{x})$, then $\mb{u}^*=\mb{0}_m$ satisfies \eqref{eqn:simpleconstraint} for all $(\mb{p},q)\in\mathcal{U}_\certfn(\mb{x})$. 

We also note that the set $~ \overline{\mathcal{U}}_\certfn(\mb{x})$ is a closed subset of the compact and convex set $~\mathcal{U}_\certfn(\mb{x})$ that is also contained in a (convex) half-space defined by the hyperplane $\mathcal{H}$, meaning $~\overline{\mathcal{U}}_\certfn(\mb{x})$ is compact and convex (as it is the intersection of convex sets).
Therefore, we can define:
\begin{equation*}
    q^{\star\star} = \max_{(\mb{p}, q) \in \overline{\mathcal{U}}_\certfn(\mb{x})} q = \max_{(\mb{p}, q) \in \overline{\mathcal{U}}_\certfn(\mb{x})} \left\langle \begin{bmatrix} \mb{0}_m \\ 1 \end{bmatrix}, \begin{bmatrix} \mb{p} \\ q \end{bmatrix} \right\rangle,
\end{equation*}
which exists as the function $\left\langle\begin{bmatrix}\mb{0}_m^\top & 1 \end{bmatrix}^\top,\cdot\right\rangle:\overline{\mathcal{U}}_\certfn(\mb{x})\to\R$ is continuous on a compact domain.
We consider the case that $q^{\star\star}>-\alpha(\certfn(\mb{x}))$, as otherwise $\mb{u}^*=\mb{0}_m$ satisfies \eqref{eqn:simpleconstraint}.

Define the projection function $\bs{\Pi}:\R^m\times\R\to\R^m$ such that for $(\mb{p},q)\in\R^m\times\R$ we have:
\begin{equation*}
    \bs{\Pi}((\mb{p},q)) = \mb{p}.
\end{equation*}
As $\bs{\Pi}$ is a linear transform, the image of the compact and convex set $\overline{\mathcal{U}}_\certfn(\mb{x})$ under $\bs{\Pi}$, denoted as $\bs{\Pi}(\overline{\mathcal{U}}_\certfn(\mb{x}))$, is compact and convex and by assumption satisfies:
\begin{equation*}
    \bs{\Pi}(\overline{\mathcal{U}}_\certfn(\mb{x})) \cap \{\mb{0}_m\} = \emptyset.
\end{equation*}
We may use the strict separating hyperplane theorem \cite{boyd2004convex} to separate the set $ \bs{\Pi}(\overline{\mathcal{U}}_\certfn(\mb{x}))$ from $\{\mb{0}_m\}$ with the hyperplane $\mathcal{H}_\beta$ with unit normal $\mb{s}\in S^{m-1}$ and offset $\beta\in\R_{++}$. This hyperplane can also be shifted to pass through the origin, given by $\mathcal{H}_0$ (with unit normal $\mb{s}\in S^{m-1}$ and offset $0$). This results in the configuration seen in Figure \ref{fig:feasibility_fig_hard_b}. 


These hyperplanes can be extended back into the ambient space $\R^m\times\R$ by defining hyperplanes $\mathcal{H}_\beta'\subset\R^m\times\R$ and $\mathcal{H}_0'\subset\R^m\times\R$ with normal vector $\mb{s}' = \begin{bmatrix} \mb{s}^\top & 0 \end{bmatrix}^\top$ and respective offsets $\beta$ and $0$. These hyperplanes serve to separate the vertical axis ${\mb{0}_m}\times\R$ from the cylinder defined by $\bs{\Pi}(\overline{\mathcal{U}}_\certfn(\mb{x}))\times\R$. Define the two following open halfspaces:
\begin{align*}
    \mathcal{H}_0^+ &\triangleq \left \{(\mb{p},q)\in\R^m\times\R ~\bigg|~  \left \langle \mb{s}', \begin{bmatrix}\mb{p} \\ q \end{bmatrix} \right \rangle = \mb{p}^\top\mb{s}  > 0 \right\} \\
    \mathcal{H}_0^- &\triangleq \left \{(\mb{p},q)\in\R^m\times\R ~\bigg|~  \left \langle \mb{s}', \begin{bmatrix}\mb{p} \\ q \end{bmatrix} \right \rangle  = \mb{p}^\top\mb{s}  < 0 \right\}
\end{align*}
We will now find a feasible input that lies anti-parallel to $\mb{s}$, i.e., we will consider inputs of the form $\mb{u}^*=-\gamma\mb{s}$ for $\gamma \geq 0$ throughout the rest of this proof. Finding a feasible input satisfying constraint \eqref{eqn:simpleconstraint} amounts to finding a $\gamma^*$ such that $q-\gamma^*\mb{p}^\top\mb{s}\leq-\alpha(\certfn(\mb{x}))$ for all $(\mb{p},q)\in\mathcal{U}_\certfn(\mb{x})$. By definition we have that for any $(\mb{p},q)\in\overline{\mathcal{U}}_\certfn(\mb{x})$:
\begin{equation*}
 \left\langle \mb{s}',\begin{bmatrix} \mb{p} \\ q \end{bmatrix}\right\rangle = \mb{s}^\top\mb{p} \geq \beta.
\end{equation*}
Let $\gamma_0 >0$ be defined as:
\begin{equation*}
\label{eqn:gamma_0}
    \gamma_0 = \frac{q^{\star\star}+\alpha(\certfn(\mb{x}))}{\beta} > 0.
\end{equation*}
This value of $\gamma_0$ implies that for any $(\mb{p},q)\in\overline{\mathcal{U}}_\certfn(\mb{x})$ we have:
\begin{equation*}
    q-\gamma_0\mb{p}^\top\mb{s} \leq q^{\star\star} -\gamma_0\beta \leq -\alpha(\certfn(\mb{x})),
\end{equation*}
Define the set:
\begin{equation*}
    \mathcal{V}_\certfn(\mb{x}, \gamma) = \left\{ (\mb{p}, q - \gamma \mb{p}^\top \mb{s}) \in \R^m \times \R ~|~ (\mb{p}, q) \in \overline{\mathcal{U}}_\certfn(\mb{x}) \right\},
\end{equation*}
and the function:
\begin{equation*}
    \psi(\gamma)  = \max_{(\mb{p},q)\in\mathcal{V}_\certfn(\mb{x}, \gamma)} q,
\end{equation*}
that satisfies $\psi(0) = q^{\star\star}$ and $\psi(\gamma_0) \leq - \alpha(\certfn(\mb{x}))$. We later show that $\psi$ is continuous, but note the intermediate value theorem implies the existence of a $\gamma^*\in(0,\gamma_0]$ with:
\begin{equation*}
    \psi(\gamma^*) = \max_{(\mb{p},q)\in\mathcal{V}_\certfn(\mb{x}, \gamma^*)} q = -\alpha(\certfn(\mb{x})).
\end{equation*}

Fixing $\mb{u}^* = -\gamma^*\mb{s}$, we have that condition \eqref{eqn:simpleconstraint} is satisfied for all points $(\mb{p},q)\in\mathcal{H}_0^+\cap~\overline{\mathcal{U}}_\certfn(\mb{x})$. Likewise, consider a point $(\mb{p},q)\in\mathcal{H}_0^+\cap(\mathcal{U}_\certfn(\mb{x})\setminus \overline{\mathcal{U}}_\certfn(\mb{x}))$, for which the condition \eqref{eqn:simpleconstraint} is satisfied with no input as $q<-\alpha(\certfn(\mb{x}))$, and note that $\mb{p}^\top\mb{s}>0$. We have that: 
\begin{equation*}
    q + \mb{p}^\top\mb{u}^* = q -\gamma^*\mb{p}^\top\mb{s} < -\alpha(\certfn(\mb{x})).
\end{equation*}
Combining these results, we have that condition \eqref{eqn:simpleconstraint} is satisfied for all $(\mb{p},q)\in\mathcal{H}_0^+\cap\mathcal{U}_\certfn(\mb{x})$. Additionally, for any $(\mb{p},q)\in\mathcal{H}_0'\cap\mathcal{U}_\certfn(\mb{x})$, $\mb{p}^\top\mb{s}=0$, as $\mb{s}'$ is normal to $\mathcal{H}_0'$. As strict separation implies that $q<-\alpha(\certfn(\mb{x}))$ for any $(\mb{p},q)\in\mathcal{H}_0'\cap\mathcal{U}_\certfn(\mb{x})$, we have $q - \gamma^*\mb{p}^\top\mb{s} < -\alpha(\certfn(\mb{x}))$ and condition \eqref{eqn:simpleconstraint} is satisfied.

\begin{figure}[t]
    \centering
    \includegraphics[scale=0.4]{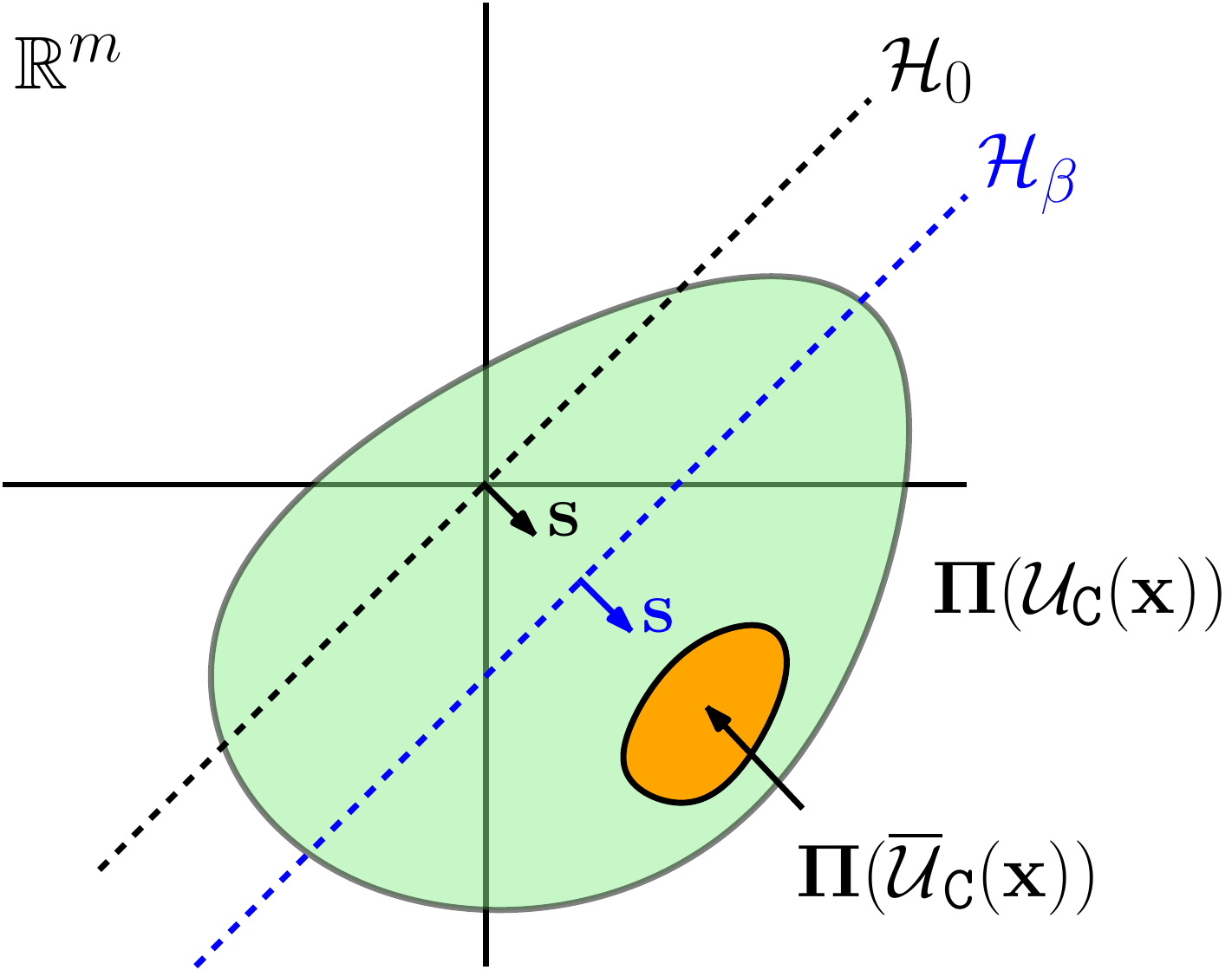}
    \caption{Projected view of Lie derivative uncertainty set used in proving feasibility of the \eqref{eqn:robust_general}. The green region corresponds to $\bs{\Pi}(\mathcal{U}_\certfn(\mb{x}))$ while the orange region highlights the set $\overline{\mathcal{U}}_\certfn(\mb{x})$ lying above the hyperplane $\mathcal{H}$. The line $\mathcal{H}_\beta$ represents the strictly separating hyperplane with normal $\mb{s}$ and offset $\beta$, while $\mathcal{H}_0$ is the hyperplane shifted to pass through the origin.}
    \label{fig:feasibility_fig_hard_b}
\end{figure}

Lastly, we must consider points $(\mb{p},q)\in\mathcal{H}_0^-\cap\mathcal{U}_\certfn(\mb{x})$. To this end, define the set:
\begin{equation*}
    \mathcal{W}_\certfn(\mb{x}, \gamma^*) = \left\{ (\mb{p}, q - \gamma^* \mb{p}^\top \mb{s}) \in \R^m \times \R ~|~ (\mb{p}, q) \in \mathcal{U}_\certfn(\mb{x}) \right\}.
\end{equation*}
This set can be interpreted as an invertible linear transformation of the set $\mathcal{U}_\certfn(\mb{x})$ as we have:
\begin{equation*}
    \begin{bmatrix} \mb{p}' \\ q' \end{bmatrix} = \begin{bmatrix} \mb{p} \\ q - \gamma^* \mb{p}^\top \mb{s} \end{bmatrix} = \begin{bmatrix} \mb{I}_{m\times m} & \mb{0}_m \\ -\gamma^* \mb{s}^\top & 1 \end{bmatrix} \begin{bmatrix} \mb{p} \\ q \end{bmatrix},
\end{equation*}
implying that $ \mathcal{W}_\certfn(\mb{x},\gamma^*)$ is compact and convex. Similarly, define the set:
\begin{align*}
    \overline{\mathcal{W}}_\certfn(\mb{x},\gamma^*) = \bigg\{(\mb{p},q)\in&~\mathcal{W}_\certfn(\mb{x},\gamma^*) ~\bigg | ~ \\ & \left\langle \begin{bmatrix} \mb{0}_m \\ 1 \end{bmatrix},\begin{bmatrix} \mb{p} \\ q \end{bmatrix} \right\rangle \geq -\alpha(\certfn(\mb{x})) \bigg\}.
\end{align*}
The set $~ \overline{\mathcal{W}}_\certfn(\mb{x},\gamma^*)$ is convex and contains points in $\mathcal{W}_\certfn(\mb{x},\gamma^*)$ that are in or above the hyperplane $\mathcal{H}$, or points that meet with equality or violate \eqref{eqn:simpleconstraint}, respectively.

By the specific choice of $\gamma^*$, there exists at least one point $\mb{v}^+\in\mathcal{H}_0^+\cap\overline{\mathcal{W}}_\certfn(\mb{x},\gamma^*)$. Furthermore, by the preceding strict separation, we have that $\mathcal{H}'_0\cap\overline{\mathcal{W}}_\certfn(\mb{x},\gamma^*) = \emptyset$. Now assume for contradiction that $\mathcal{H}_0^- \cap \overline{\mathcal{W}}_\certfn(\mb{x},\gamma^*) \neq \emptyset$, or that there exists $\mb{v}^-\in \mathcal{H}_0^-  \cap \overline{\mathcal{W}}_\certfn(\mb{x},\gamma^*)$. As the set $~\overline{\mathcal{W}}_\certfn(\mb{x},\gamma^*)$ is convex and $\langle\mb{s}',\mb{v}^+\rangle>0$ and $\langle\mb{s}',\mb{v}^-\rangle<0$, there is a $\lambda^* \in (0, 1)$ satisfying
\begin{equation*}
    \langle\mb{s}' , (1 - \lambda^*) \mb{v}^+ + \lambda^* \mb{v}^-\rangle = 0,
\end{equation*}
implying $(1 - \lambda^*) \mb{v}^+ + \lambda^* \mb{v}^-\in\mathcal{H}'_0\cap\overline{\mathcal{W}}_\certfn(\mb{x},\gamma^*)$. This contradicts the fact $\mathcal{H}'_0\cap\overline{\mathcal{W}}_\certfn(\mb{x},\gamma^*) = \emptyset$, implying that $\mathcal{H}^-_0\cap\overline{\mathcal{W}}_\certfn(\mb{x},\gamma^*) = \emptyset$, or that condition \eqref{eqn:simpleconstraint} is satisfied for all points $(\mb{p}, q) \in\mathcal{H}_0^{-}\cap \mathcal{U}_\certfn(\mb{x})$. Combining all previous results, condition \eqref{eqn:simpleconstraint} is satisfied for all points $(\mb{p},q)\in\mathcal{U}_\certfn(\mb{x})$ for the input $\mb{u}^* = -\gamma^*\mb{s}$, ensuring feasibility.

\begin{figure*}[b]
    \centering
    \includegraphics[width=\textwidth]{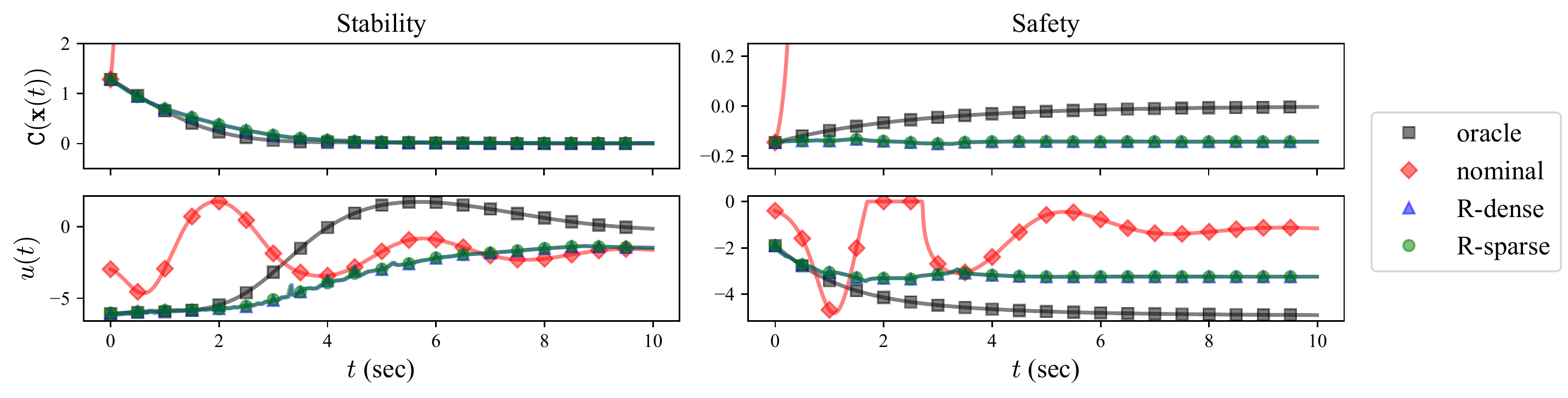}
    \caption{The value of the certificate function $\certfn(\mathbf x)$ (top) and input $\mathbf u$ (bottom) over time for the stability (left) and safety (right) experiments. 
    We compare the behavior of an oracle controller which has access to the true dynamics (black square), a nominal controller which assumes that the estimated dynamics are true (red diamond), a robust controller using dense input data (blue triangle), and a robust controller using sparse input data (green circle).
    For stability, the desired behavior for $\certfn(\mathbf x)$ is to converge towards $0$, while for safety it is to remain nonpositive.}
    \label{fig:sim_results}
\end{figure*}

Finally, we show that $\psi$ is continuous. For $\gamma, \gamma' > 0$ and $(\mb{p}, q) \in \overline{\mathcal{U}}_\certfn(\mb{x})$, we have:
\begin{align*}
    &\min_{(\mb{p}', q') \in \mathcal{V}_{\certfn}(\mb{x}, \gamma')} \left\Vert \begin{bmatrix} \mb{p} \\ q - \gamma \mb{p}^\top\mb{s} \end{bmatrix} - \begin{bmatrix} \mb{p}' \\ q' \end{bmatrix} \right\Vert_2\\
    &~\leq \left\Vert \begin{bmatrix} \mb{p} \\ q - \gamma \mb{p}^\top\mb{s} \end{bmatrix} - \begin{bmatrix} \mb{p} \\ q - \gamma' \mb{p}^\top \mb{s} \end{bmatrix} \right\Vert_2 \leq \mb{p}^\top\mb{s} \vert \gamma - \gamma' \vert,
\end{align*}
and similarly, for $(\mb{p}', q') \in \overline{\mathcal{U}}_\certfn(\mb{x})$, we have:
\begin{equation*}
    \min_{(\mb{p}, q) \in \mathcal{V}_{\certfn}(\mb{x}, \gamma)} \left\Vert \begin{bmatrix} \mb{p} \\ q \end{bmatrix} - \begin{bmatrix} \mb{p}' \\ q' - \gamma' (\mb{p}')^\top \mb{s} \end{bmatrix} \right\Vert_2 \leq (\mb{p}')^\top \mb{s} | \gamma - \gamma' |.
\end{equation*}
Therefore, the Hausdorff distance between $\mathcal{V}_\certfn(\mb{x}, \gamma)$ and $\mathcal{V}_\certfn(\mb{x}, \gamma')$ is bounded by:
\begin{equation*}
    d_H(  \mathcal{V}_\certfn(\mb{x}, \gamma),   \mathcal{V}_\certfn(\mb{x}, \gamma')) \leq \left( \max_{(\mb{p}, q) \in \overline{\mathcal{U}}_\certfn(\mb{x})} \mb{p}^\top\mb{s} \right) \vert \gamma - \gamma' \vert,
\end{equation*}
implying $\mathcal{V}_\certfn(\mb{x}, \gamma)$ is Lipschitz continuous (with respect to the Hausdorff metric) as a function of $\gamma$. The support function of a nonempty, compact, and convex set $\mathcal{A}\subset\R^{m}\times\R$ given by $h_{\mathcal{A}}:\R^{m+1}\to\R$ is defined as:
\begin{equation*}
   h_{\mathcal{A}}(\mb{v}) =  \max_{(\mb{p},q)\in\mathcal{A}} \left\langle \mb{v}, \begin{bmatrix}\mb{p} \\ q \end{bmatrix} \right\rangle
\end{equation*}
Recalling $\mb{n} = \begin{bmatrix} \mb{0}_m^\top & 1 \end{bmatrix}^\top$, the Hausdorff distance between two nonempty, compact, and convex sets $\mathcal{A},\mathcal{B}\subset\R^{m}\times\R$ satisfies:
\begin{equation*}
    d_H(\mathcal{A},\mathcal{B}) = \max_{\mb{v}\in S^m}\vert h_\mathcal{A}(\mb{v}) - h_\mathcal{B}(\mb{v}) \vert \geq \vert h_\mathcal{A}(\mb{n})-h_\mathcal{B}(\mb{n}) \vert.
\end{equation*}
Noting that $\psi$ can be expressed in terms of the support function as a composition, with:
\begin{equation*}
    \psi(\gamma) =  \max_{(\mb{p},q)\in\mathcal{V}_\certfn(\mb{x}, \gamma)} q = h_{\mathcal{V}_\certfn(\mb{x}, \gamma)}(\mb{n}),
\end{equation*}
implies $\psi$ is a continuous function.
\end{proof}

\section{Simulation}
\label{sec:sim}

To demonstrate the capabilities of the proposed controller, we run simulated experiments in the setting of an 
inverted pendulum, described by the system model:
\begin{equation}
\label{eqn:pendmodel}
    \frac{\mathrm{d}}{\mathrm{d}t}\begin{bmatrix}\theta \\ \dot{\theta} \end{bmatrix} = \begin{bmatrix} \dot{\theta} \\ \frac{\hat{g}}{\hat{\ell}}\sin(\theta) \end{bmatrix} + \begin{bmatrix} 0 \\ \frac{1}{\hat{m}\hat{\ell}^2} \end{bmatrix}u 
\end{equation}
with state $\mathbf x = [\theta, \dot\theta]^\top\in\R^2$, gravitational acceleration estimate $\hat{g} =10$, length estimate $\hat{\ell} = 0.63$, and mass estimate $\hat{m} = 0.63$. We assume that the unknown true system has modified inverted pendulum dynamics given by:
\begin{equation}
\label{eqn:pendtrue}
    \frac{\mathrm{d}}{\mathrm{d}t}\begin{bmatrix}\theta \\ \dot{\theta} \end{bmatrix} = \begin{bmatrix} \dot{\theta} \\ \frac{g}{\ell}\sin(\theta) \end{bmatrix} + \begin{bmatrix} 0 \\ \frac{1-0.75\exp(-\theta^2)}{m\ell^2} \end{bmatrix}u 
\end{equation}
with gravitational acceleration $g = 10$, length $\ell = 0.7$, and mass $m = 0.7$. We note the inclusion of a state-dependent input gain given by
$1 - 0.75 \exp(-\theta^2)$, which attenuates the input most significantly when the pendulum is vertical.

Consider the functions $V:\R^2\to\R_{+}$ and $h:\R^2\to\R$, given by $V(\mb{x}) = \mb{x}^\top\mb{Px}$ and $h(\mb{x}) =  \mb{x}^\top\mb{P}\mb{x} - c$ with
\begin{equation}
    \mb{P} = \begin{bmatrix} \sqrt{3} & 1 \\ 1 & \sqrt{3} \end{bmatrix}
\end{equation}
and a constant $c = 0.2$. Noting that both the system model \eqref{eqn:pendmodel} and the true system \eqref{eqn:pendtrue} are feedback linearizable, $V$ and $h$ satisfy the CCF condition \eqref{eqn:ccf} for the comparison function $\alpha(r)=\lambda_{\min}(\mb{Q}) r / \lambda_{\max}(\mb{P})$, for both the model and the true system (implying Assumption \ref{ass:wellposed} is met). In particular, $V$ may serve as a CLF, and $h$ as a CBF.

We explore data-driven stability and safety with this system for different methods of gathering training data.
We compare the robust data-driven controller with an oracle controller, which knows the true model, and a nominal controller, which treats the estimated model as if it were true, both specified via the \eqref{eqn:CCF-QP}. In each setting, the system model underestimates the pendulum mass and length and assumes that the input gain is independent of the state. As a result, the Lipschitz constants of the errors can be bounded by $\Lips_{\errf}=g| \ell - \hat \ell|/( \ell \hat \ell)$ and $\Lips_{\errg}=0.75\sqrt{2}\exp(-\frac{1}{2})/(m\ell^2)$. 

For the stability experiment, we generate data sets by gridding the state and input spaces. We consider $\theta$ in the interval $[0, 1]$ and $\dot{\theta}$ in the interval $[-0.25, 0.25]$, with grid sizes $\epsilon_\theta = \epsilon_{\dot{\theta}} = \frac{1}{40}$. We generate a \textit{sparse} data set by considering $u \in \{ -5, -1 \}$ and a \textit{dense} data set by considering $u \in \{ -5, -3, -1, 1, 3, 5 \}$.
We set the regularizing controller $\mb{k}_d$ to be a feedback linearizing controller designed using the estimated dynamics and linear gains $k_p = 1/2$ and $k_d = \sqrt{3}/2$.
The system is simulated from the initial condition $\mb{x} = [0.8, 0.1]$ for $10$ seconds, with control inputs specified at $100$ Hz. 
For the safety experiment, we consider a similar pair of sparse and dense data sets with $\theta$ in the interval $[0, 0.25]$ and $\dot{\theta}$ in the same interval as the stability experiment, the same grid sizes, and the same sets of control inputs. The system is simulated with no regularizing controller $\mb{k}_d\equiv0$ from the initial condition $\mb{x} = [0.1, 0.1]$ for the same amount of time and the same control input frequency. 

The results of the simulations may be seen in Figure \ref{fig:sim_results}. In both experiments, we see that the nominal controller fails to achieve the specified objective, while the oracle controller succeeds. Furthermore, we see that the data-driven controllers perform nearly identically for both the sparse and dense input data sets. This similarity indicates that greater variety in input directions and coverage of the input space by the data set are not needed to achieve satisfactory closed-loop behavior.
\section{Conclusion}
We propose a novel approach for robust data-driven control synthesis under model uncertainty
and show that good performance can be achieved when training data provides sufficiently dense coverage of the state space with targeted excitation in input directions. This approach is well positioned to make progress on both theoretical and application-oriented problems at the intersection nonlinear control and non-parametric machine learning. Future work includes investigating strategies for data segmentation to enable efficient computation for closed-loop control and understanding continuity and recursive feasibility of the controller.


\bibliographystyle{IEEEtran} 
\bibliography{main}

\end{document}